\documentclass{amsart}
\usepackage{amsmath,amssymb}
\usepackage{comment}
\usepackage{graphicx}
\usepackage{dcolumn}
\usepackage{bm}
\usepackage{amsthm} %
\usepackage{amscd} %
\usepackage{mathrsfs} 
\usepackage{color} 
\newtheorem{prop}{Proposition}
\newtheorem{coro}[prop]{Corollary}

\theoremstyle{definition}

\theoremstyle{remark}
\newtheorem{remark}{Remark}
\numberwithin{equation}{section}

\begin{document}
\newcommand{\cstar}{{C}^{\ast}}%
\newcommand{\wstar}{{W}^{\ast}}%
\newcommand{\sinc}{{\rm{sinc}}}
%
\newcommand{\R}{{\mathbb{R}}}%
\newcommand{\Z}{{\mathbb{Z}}}%
\newcommand{\CC}{{\mathbb{C}}}%
\newcommand{\NN}{{\mathbb{N}}}%
\newcommand{\nonum}{\nonumber}%
%
\newcommand{\al}{\alpha}
\newcommand{\alz}{\alpha_{z}}%
\newcommand{\ome}{\omega}%
\newcommand{\vp}{\varphi}
\newcommand{\lam}{\lambda}%
\newcommand{\Lam}{\Lambda}%
\newcommand{\cicr}{c_i^{\,\ast}}%
\newcommand{\ci}{c_i}%
\newcommand{\ctwocr}{c_2^{\,\ast}}%
\newcommand{\ctwo}{c_2}%
\newcommand{\cjcr}{c_j^{\,\ast}}%
\newcommand{\cj}{c_j}%
\newcommand{\clcr}{c_l^{\,\ast}}%
\newcommand{\cl}{c_l}%
\newcommand{\cjkcr}{c_{j+k}^{\,\ast}}%
\newcommand{\cjk}{c_{j+k}}%
\newcommand{\cijcr}{c_{i+j}^{\,\ast}}%
\newcommand{\cij}{c_{i+j}}%
\newcommand{\cjicr}{c_{j+i}^{\,\ast}}%
\newcommand{\cji}{c_{j+i}}%
\newcommand{\cjonecr}{c_{j+1}^{\,\ast}}%
\newcommand{\cjone}{c_{j+1}}%
\newcommand{\cmcr}{c_m^{\,\ast}}%
\newcommand{\cm}{c_m}%
\newcommand{\cncr}{c_n^{\,\ast}}%
\newcommand{\cn}{c_n}%
\newcommand{\hck}{\hat{c}_k}%
\newcommand{\hckcr}{\hat{c}_k^{\ast}}%
\newcommand{\hcl}{\hat{c}_l}%
\newcommand{\hclcr}{\hat{c}_l^{\ast}}%
\newcommand{\Nk}{\hat{N}_k}%

\newcommand{\vpweak}{\tilde{\vp}}

\newcommand{\I}{{\mathrm{I}}}%
\newcommand{\J}{{\mathrm{J}}}%
\newcommand{\K}{{\mathrm{K}}}%
\newcommand{\X}{{\mathrm{X}}}%
\newcommand{\Y}{{\mathrm{Y}}}%
\newcommand{\V}{{\mathrm{V}}}%
\newcommand{\Ic}{\I^{c}}%
\newcommand{\Jc}{\J^{c}}%
\newcommand{\Lamc}{\Lam^{c}}%
\newcommand{\KN}{\K_N}%
\newcommand{\Al}{\mathfrak{A}}%
\newcommand{\uniAl}{\hat{\Al}}%
\newcommand{\Als}{\Al^{S}}%
\newcommand{\Thes}{\Theta^{S}}%
\newcommand{\Fl}{\Al^F}%
\newcommand{\taus}{\tau^{S}}%
\newcommand{\tausgamma}{\tau^{S, \gamma}}%
\newcommand{\tauf}{\tau^F}%
\newcommand{\alps}{\alpha^{S}}%
\newcommand{\alpf}{\alpha^F}%
\newcommand{\Flpi}{\Fl_{[-\pi, \pi]}}%
\newcommand{\core}{\Fl_{\circ}}%
\newcommand{\cores}{\Als_{\circ}}%
\newcommand{\FlI}{\Fl(\I)}%
\newcommand{\AlsI}{\Als(\I)}%
\newcommand{\FlIc}{\Fl(\Ic)}%
\newcommand{\AlsIc}{\Als(\Ic)}%
\newcommand{\FlN}{\Fl_{[N]}}%
\newcommand{\AlsN}{\Als_{[N]}}%
\newcommand{\AlsNx}{\Als_{[N]+x}}%
\newcommand{\FlM}{\Fl_{[M]}}%
\newcommand{\AlsM}{\Als_{[M]}}%
\newcommand{\AlsidouM}{\Als_{[M]+\tint}}%
\newcommand{\idouM}{[M]+\tint}%
\newcommand{\FlLam}{\Fl(\Lam)}%
\newcommand{\AlsLam}{\Als(\Lam)}%
\newcommand{\FlLamc}{\Fl(\Lamc)}%
\newcommand{\AlsLamc}{\Als(\Lamc)}%

\newcommand{\AlJ}{\Als(\J)}%
\newcommand{\FlJ}{\Fl(\J)}%
%
\newcommand{\Ale}{\Al_{+}}%
\newcommand{\Alo}{\Al_{-}}%
\newcommand{\Fle}{\Fl_{+}}%
\newcommand{\Flo}{\Fl_{-}}%
\newcommand{\Alse}{\Als_{+}}%
\newcommand{\Also}{\Als_{-}}%
\newcommand{\Ks}{\mathcal{K}}%
\newcommand{\Ul}{{\mathcal{U}}}%
\newcommand{\Bl}{{\mathcal{B}}}%
\newcommand{\tr}{{\rm{tr}}}
\newcommand{\tint}{t_{\rm{int}}}%
\newcommand{\tdec}{t_{\rm{dec}}}%
\newcommand{\AUT}{{\rm{Aut}}}%
\newcommand{\AUTf}{\AUT(\Fl)}%
\newcommand{\AUTs}{\AUT(\Als)}%
\newcommand{\AUTeven}{\AUT(\Fle)}%
\newcommand{\id}{{\rm{id}}}%
\newcommand{\coree}{{\core}_+}%
\newcommand{\coreo}{{\core}_-}%
\newcommand{\coreeo}{{\core}_\pm}%
\newcommand{\coreoe}{{\core}_\mp}%
\newcommand{\Ap}{A_{+}}%
\newcommand{\Am}{A_{-}}%
\newcommand{\Apm}{A_{\pm}}%
\newcommand{\Bp}{B_{+}}%
\newcommand{\Bm}{B_{-}}%
\newcommand{\Bpm}{B_{\pm}}%
\newcommand{\ineps}{\in^{\varepsilon}}%
\newcommand{\inDoubeps}{\in^{2\varepsilon}}%
\newcommand{\spec}{{\rm{sp}}}%
\newcommand{\Hinf}{\tilde{H}_{\infty}}%
\newcommand{\emacs}{\texttt{emacs}}
\newcommand{\aarrow}{\stackrel{\alpha}\rightarrow}
\newcommand{\barrow}{\stackrel{\beta}\rightarrow}
\newcommand{\epsINC}{\stackrel{\varepsilon}\in}

\title{Continuous extension of the  discrete shift translations 
 	on  one-dimensional quantum   lattice systems }

\author{Hajime Moriya}
\address{Faculty of Mechanical Engineering, Institute of Science and 
Engineering, Kanazawa University, Kakuma, Kanazawa 920-1192, Japan}
\curraddr{}
\email{hmoriya4@se.kanazawa-u.ac.jp}
\author{Heide Narnhofer}
\address{Faculty of Physics, University of Vienna, 
 Boltzmanngasse 5, A-1090 Vienna, Austria}
\curraddr{}
\email{heide.narnhofer@univie.ac.at}


\date{}

\dedicatory{}

\begin{abstract}
We study   the continuous  extension of 
discrete shift translations on
 one-dimensional  quantum  lattice systems. 
 We explore  a specific  example   provided   by  
a  quasi-free $\cstar$-flow on 
 the  one-dimensional  fermion lattice system, 
comparing it 
 with its original   discrete shift-translations.
Specifically, we demonstrate 
  its explicit dynamical  formula, 
which reveals violations of causality and locality. 
Furthermore, we prove   that  
this  quasi-free $\cstar$-flow on 
 the  one-dimensional  fermion lattice system
 cannot be extended to 
 the one-dimensional quantum spin lattice system via the 
   Jordan-Wigner transformation.
\end{abstract}

\maketitle

\section{Introduction}
\label{sec:INT}
We  investigate the continuous  extension  
 of   discrete shift translations 
 on  the one-dimensional   
  fermion lattice system  and those  on the 
 one-dimensional quantum spin lattice system 
 to  $\cstar$-flows.  
To provide context for our investigation, we 
 review    some  relevant previous  results. 
 In   \cite{RWW},  it has  been shown  that   the  continuous 
 extension of the  shift-translation automorphism group on the 
 one-dimensional quantum spin lattice system   
cannot be achieved  by  local Hamiltonians  satisfying certain 
 locality condition.  
This non-existence statement  (Corollary 6.1  of \cite{RWW})
 is a generalization   of  the non-implementability of the shift translation  
by  finite quantum circuits   \cite{GRO}.
We also  refer to  \cite{YIN}, where  the implementation of shift translations  by Hamiltonians on finite quantum-spin chains is studied in detail.

Our ultimate objective is to demonstrate all $\cstar$-flows (strongly continuous one-parameter groups of automorphisms) that implement continuous 
 shift translations on one-dimensional quantum systems, or to disprove the 
existence of such $\cstar$-flows.
In this article, we narrow our focus to a 
 quasi-free extension of discrete shift translations 
 on the one-dimensional fermion lattice system.
Due to the  aforementioned no-go statements, this quasi-free $\cstar$-flow 
  lacks  finite   propagation speed and  
 even   the  locality  that implies 
  the Lieb-Robinson bound. 
 We  examine dynamical properties of this quasi-free $\cstar$-flow
in comparison with its original   discrete shift-translation automorphism group
  on the  fermion system. Additionally, we demonstrate that 
 the quasi-free $\cstar$-flow on the  fermion system  is not  extendable  
to the one-dimensional quantum spin lattice system; 
this constitutes a form of no-go theorem.

From the mathematical perspective, 
the pursuit of continuous extensions of a discrete automorphism group (or a single automorphism) finds its origins in Kuiper's seminal work \cite{KUI}, 
 which explores the topological structure of unitaries on 
 infinite-dimensional Hilbert spaces. 
Previous studies  \cite{GRO} \cite{RWW} mentioned before 
 have examined  the topological structure 
 of the  discrete shift translations on 
one-dimensional quantum spin systems 
 using Index theory, and  
 recent research  \cite{KATO} has  further 
developed the topological aspects shown in \cite{GRO}.  
According to  \cite{POPA}, which  generalizes  Kuiper's work to  general 
 von Neumann algebras, 
the $\sigma$-weakly continuous   extension of    
the  discrete shift translations  
 on von Neumann algebras (generated by the CAR algebra) obviously exists.

Finally, let us   mention  the work \cite{JONES}, which    
  deals with a somewhat similar 
 extension problem  for the shift translations on the quantum spin chain.
The work \cite{JONES} explores
   the  continuous limit of  shift translations on 
 the one-dimensional quantum-spin lattice 
 system   by letting 
  the spacing of the unit lattice  approach  zero, whereas 
  we always  consider   the given   quantum lattice system as fixed.
Our work and  \cite{JONES} tackle different issues.

\section{Mathematical formulation}
\label{sec:MF}
\subsection{Quantum  systems}
\label{subsec:QS}

We  introduce  the spinless fermion lattice system
 and the spin-one-half  quantum spin lattice system   on $\Z$ as 
 quasi-local $\cstar$-systems.
We refer to
 \cite{BR} for the   reference of
 $\cstar$-algebraic quantum statistical mechanics.

  For a  subset $\I\subset \Z$,
  $|\I|$ denotes  the number of sites
 in $\I$.  If $\I\subset\Z$  has a finite   $|\I|<\infty$,  
  we shall denote  $\I \Subset \Z$.
Let $\ci$ and $\cicr$ denote the annihilation operator
and the creation operator of a spinless fermion sitting at $i\in \Z$, respectively. These  obey the  canonical 
anti-commutation  relations (CARs):
\begin{align}
\label{eq:CAR}
\{ \cicr, \cjcr \}=\{ \ci, \cj \}=0,\quad \{ \cicr, \cj \}&=\delta_{i,j}\, 1, 
 \quad i, j\in \Z.
\end{align}
For each $\I\Subset\Z$,
define  the finite system  $\FlI$ 
by the  $\ast$-algebra 
 generated by $\{\cicr, \, \ci\, ;\;i\in \I\}$.
 It is isomorphic to ${\mathrm M}_{2^{|\I|}}(\CC)$, the algebra of all
 $2^{|\I|}\times 2^{|\I|}$ complex matrices.
 For  $\I \subset \J\Subset \Z$, $\FlI$ is naturally embedded into $\FlJ$ as a
subalgebra. 
Let us take 
\begin{equation}
\label{eq:CARloc}
\core:=\bigcup_{\I \Subset \Z }\FlI.
\end{equation}
By  taking  the  norm (operator norm) completion of  $\core$
we obtain  a  $\cstar$-algebra  $\Fl$   called
   the CAR algebra.
The  dense $\ast$-subalgebra $\core$ in $\Fl$ is 
called the local algebra of $\Fl$.

Let $\Theta$ denote the   automorphism on the $\cstar$-algebra
 $\Fl$  determined  by
 \begin{equation}
\label{eq:CARTheta}
\Theta(\ci)=-\ci, \quad \Theta(\cicr)=-\cicr,\quad i\in \Z.
\end{equation}
 As the automorphism $\Theta\in \AUTf$
 satisfies $\Theta\circ \Theta=\id \in \AUTf$,  
it provides  the  graded structure with  $\Fl$:
\begin{align}
\label{eq:CAREO}
 \Fle &:= \{A \in \Fl \; \bigl| \;   \Theta(A)=A  \}, \quad
 \Flo := \{A \in \Fl  \; \bigl| \;  \Theta(A)=-A  \}, \nonum \\ 
\Fl&=\Fle+\Flo.
\end{align}
By   the  graded structure of $\Fl$,  
 any element $A\in \Fl$  has its unique even-odd decomposition:
\begin{align}
\label{eq:decompCAREO}
A=\Ap+\Am, \ 
\Ap:=\frac{1}{2}(A+\Theta(A))\in  \Fle,\ 
\Am:=\frac{1}{2}(A-\Theta(A))\in  \Flo.
\end{align}

Next,  we will introduce   the  quantum spin lattice system $\Als$  
 via the fermion lattice system
  $\Fl$  using the $\cstar$-algebraic formalism given in  \cite{AR-XY}.
This  somewhat  roundabout method provides a rigorous  frame  to deal with 
 the Jordan-Wigner transformation 
 connecting   the fermion lattice system to   the quantum 
spin lattice system. 
Let  
\begin{equation}
\label{eq:Thetami}
\Theta_{-}(\ci^{\sharp})=
\begin{cases}
\ci^{\sharp} & i \ge 1 \\
-\ci^{\sharp} & i \le  0
\end{cases}, 
\end{equation}
where $\ci^{\sharp}$ denotes  $\ci$ or $\ci^{\ast}$.
Let $\uniAl$  denote  the $\cstar$-crossed product of $\Fl$
by the $\Z_2$-action $\Theta_{-}\in\AUTf$.
 Concretely,  $\uniAl$ is generated by the algebra  $\Fl$ and the
 additional element $T$ satisfying the following algebraic relations:
\begin{equation}
\label{eq:Th}
T=T^{\ast}, \ T^2=1,\
T AT =\Theta_{-}(A)\ \text{for}\  A\in \Fl.
\end{equation}
The grading $\Theta$ on $\Fl$ 
 is uniquely extended to  $\uniAl$ by determining 
 \begin{equation}
\label{eq:ThetaT}
\Theta(T)=T.
\end{equation}
Obviously, $\uniAl$ is decomposed into the following four  
linear subspaces:
\begin{equation}
\label{eq:uniAldec}
\uniAl=\Fle+T\Fle+\Flo+T\Flo.
\end{equation}
By definition,  
  $\uniAl$ includes both the fermion lattice system $\Fl(=\Fle+\Flo)$ and 
   the quantum spin lattice system $\Als$ 
 defined by
\begin{equation}
\label{eq:Als}
\Als\equiv \Fle+T\Flo.
\end{equation}
  The spin operators 
 $\{\sigma_j^{x}, \sigma_j^{y}, \sigma_j^{z}\in {\mathrm M}_2(\CC) \}$
  at each  $j\in\Z$
 are explicitly given  by 
\begin{align}
\label{eq:spinXYZ}
\sigma_j^{z}:=2\cjcr \cj -1,\quad
\sigma_j^{x}:=TS^{(j)} (\cj+ \cjcr),\quad
\sigma_j^{y}:=iTS^{(j)} (\cj- \cjcr),
\end{align}
 with 
\begin{align}
\label{eq:spin-product}
S^{(j)}:=
\begin{cases}
\sigma_1^{z}\cdots \sigma_{j-1}^{z}   &  j \ge  2 \\
    1  &  j=1 \\
    \sigma_0^{z}\cdots \sigma_{j}^{z} & j \le  0
\end{cases}.
\end{align}
 By  the above definition, using  the CARs \eqref{eq:CAR},  we see that 
  the set 
  $\{\sigma_j^{x}, \sigma_j^{y}, \sigma_j^{z}\in \uniAl; \; j\in\Z \}$ 
 generates  the  quantum spin algebra  $\Als$  on $\Z$
 given in \eqref{eq:Als} as  a subalgebra of $\uniAl$.
Let $\Thes$ denote the   automorphism on the $\cstar$-algebra
 $\Als$  determined  by
 \begin{equation}
\label{eq:spinTheta}
\Thes(\sigma_j^{x})=-\sigma_j^{x}, \ 
\Thes(\sigma_j^{y})=-\sigma_j^{y}, \ 
\Thes(\sigma_j^{z})=\sigma_j^{z}, \quad  i\in \Z.
\end{equation}
 As  $\Thes\circ \Thes=\id \in \AUTs$,  $\Thes$  
provides  the  graded structure with  $\Als$ as
\begin{align}
\label{eq:SPINEO}
 \Alse &:= \{A \in \Als \; \bigl| \;   \Thes(A)=A  \},\quad
 \Also := \{A \in \Als  \; \bigl| \;  \Thes(A)=-A  \}, \nonum\\
\Als&=\Alse+\Also.
\end{align}
We  easily see that  
  $\Alse$ is generated by  
  the set $\{\sigma_j^{z}\in \Als; \; j\in\Z \}$ and even monomials of  
  $\{\sigma_j^{x}, \sigma_j^{y}; \; j\in\Z \}$, and 
 $\Also$ is generated by  
  odd  monomials of  
  $\{\sigma_j^{x}, \sigma_j^{y}\in \Als; \; j\in\Z \}$.
The following identities   give the location of the fermion algebra and the 
 quantum spin algebra in $\uniAl$:
\begin{equation}
\label{eq:Alseven-odd-coin}
\Alse= \Fle, \quad \Also= T\Flo.
\end{equation}
For each $\I\Subset\Z$,
define  the subalgebra  $\AlsI$ of $\Als$
by the  finite-dimensional algebra 
 generated by $\{\sigma_j^{x}, \sigma_j^{y}, \sigma_j^{z} ;\;j\in \I\}$.
Let 
\begin{equation}
\label{eq:spinloc}
\cores:=\bigcup_{\I \Subset \Z }\AlsI.
\end{equation}
This  norm dense $\ast$-subalgebra  in   $\Als$ 
 is   called the local algebra of $\Als$.

\subsection{Discrete and continuous shift translations}
\label{subsec:SHIFT}
We  first  define  the discrete shift translations  on 
 $\Fl$ and then   those on $\Als$.
For each $i\in \Z$, let
\begin{align}
\label{eq:shift-f}
\tauf_{i}(\cj)=\cji, \
\tauf_{i}(\cjcr)=\cjicr, \quad j\in \Z.
\end{align}
The above formulas 
 determine 
 the shift-translation automorphism group $\{\tauf_i\in \AUTf,\;i\in \Z\}$
  on $\Fl$.
For each $i\in \Z$,  let  
\begin{align}
\label{eq:shift-s}
\taus_{i}(\sigma_j^{a})=
 \sigma_{j+i}^{a}, 
\quad a\in \{x,y,z\}, 
\quad j\in \Z.
\end{align}
Then 
$\{\taus_i \in \AUTs,\;i\in \Z\}$
 gives  the  shift-translation automorphism  group  on $\Als$.
Our purpose is to extend these  
$\Z$-shift-translation automorphism  groups to 
strongly continuous 
one-parameter group of automorphisms, i.e. $\cstar$-flows.

In the following,  we   construct a
   continuous  shift-translation group on $\Fl$. 
 First, note that our fermion lattice 
 system $\Fl$  is identical to  the CAR algebra
over  the complex Hilbert space $l^2(\Z)$.
This identification is done 
 by setting   $c_j=c(\chi_j)$ for  $j\in \Z$, where 
 $\chi_j \in l^2(\Z)$ denotes 
 the indicator determined  by  
\begin{align}
\label{eq:indicator}
\chi_j(k):=\delta_{jk},\ (k\in \Z).
\end{align}
For each $f=(f_j)_{j\in \Z}\in l^2(\Z)$, 
 define
\begin{align}
\label{eq:CARf}
c(f):=\sum_{j\in\Z}f_j \cj \in \Fl,
\quad c^{\ast}(f)
:=\sum_{j\in\Z}f_j \cjcr\in \Fl.
\end{align}
 $\Fl$ is  also isomorphic  to the CAR algebra over
 $L^2[-\pi, \pi]$ as $\cstar$-algebras
 by the isometry of  $l^2(\Z)$
 and $L^2[-\pi, \pi]$ as   Hilbert spaces.
 Here, the closed interval   $[-\pi, \pi]$ can  be 
 interpreted  as the range of the possible  momentum  of 
 each  fermion particle.
The  connection between 
$f=(f_j)_{j\in \Z}\in l^2(\Z)$ and $\tilde{f}(k)\in L^2[-\pi, \pi]$
 is given by  the Fourier transformation as follows:
\begin{align}
\label{eq:Fourier}
\tilde{f}(k):=\sum_{j\in\Z}f_j e^{ijk}, \quad k\in [-\pi, \pi],\nonumber \\
\quad f_j=\frac{1}{2\pi}\int_{-\pi}^{\pi}
\tilde{f}(k) e^{-ijk}\,dk  \quad j\in \Z.
\end{align}
Hence for each $j\in \Z$
\begin{align}
\label{eq:chiFourier}
\widetilde{{\chi_j}}(k)= e^{ijk}, \quad k\in [-\pi, \pi].
\end{align}
The one-step (right) shift translation  $U_1$   
 on the  Hilbert space  $l^2(\Z)$
 is given by 
\begin{align}
\label{eq:Ufone}
(U_1 f)_j:=f_{j-1}, \quad j\in \Z,
\quad \text{for}\ f=(f_j)_{j\in \Z}\in l^2(\Z).  
\end{align}
The   $\Z$-group of unitaries 
 $\{U_i\in  \Ul(L^2[-\pi, \pi]);\; i\in \Z\}$ 
is given by  
\begin{equation*}
 U_0=I,\ 
 U_k:=	\underbrace{U_1\cdots U_1}_{k-{\rm{times}}}, 
\  U_{-k}:={U_{k}}^{-1}
 \  (k\in \NN).   
\end{equation*}
By the Fourier transformation,  
$U_1 \in \Ul(l^2(\Z))$ 
  has the following expression  on   $L^2[-\pi, \pi]$ 
\begin{align}
\label{eq:Utilde-one}
\widetilde{U_1 f}(k)=
e^{ik}\tilde{f}(k), \quad
k\in [-\pi, \pi].
\end{align}
By the  
  second quantization procedure, 
 the   $\Z$-group of unitaries 
 $\{U_i\in  \Ul(L^2[-\pi, \pi]);\; i\in \Z\}$ 
   generates  
 the  discrete-shift translation automorphism group 
$\{\tauf_i\in \AUTf,\;i\in \Z\}$ given in \eqref{eq:shift-f}. 
To  interpolate   
  $\{U_i\in  \Ul(L^2[-\pi, \pi]);\; i\in \Z\}$, 
  we  define unitaries   $U_t$ ($t\in\R$)   
 on $L^2[-\pi, \pi]$ as 
\begin{align}
\label{eq:INTERt}
\widetilde{U_t f}(k):=
e^{itk}\tilde{f}(k), \quad k\in [-\pi, \pi]. 
\end{align}
Now we obtain  a   strongly continuous one-parameter
 group of unitaries: 
\begin{equation}
\label{eq:UtPi}
\{U_t\in  \Ul(L^2[-\pi, \pi]);\; t\in \R\}, 
\end{equation}
  equivalently, 
\begin{equation}
\label{eq:UtZ}
\{U_t\in  \Ul(l^2(\Z);\; t\in \R\}.
\end{equation}
By Stone's theorem \cite{Stone},  
there exists a unique  self-adjoint operator 
 $\bar{h}$ on $L^2[-\pi, \pi]$ such that 
\begin{equation}
\label{eq:Utgenerator}
U_t=e^{it \bar{h}}, \quad t\in \R.
\end{equation}
 This  one-parameter
 group of unitaries generates  a  strongly continuous
one parameter  group of quasi-free (Bogoliubov) automorphisms
on $\Fl$: 
\begin{equation}
\label{eq:CARquasifree}
\{\tauf_t \in  \AUTf,\;t\in \R\}.   
\end{equation}
By  construction, it is  a desired $\cstar$-flow 
that  extends  the  discrete 
 shift-automorphism group   $\{\tauf_i\in \AUTf,\;i\in \Z\}$
 from the discrete parameter $\Z$ to the continuous parameter $\R$.
We  see that the continuous shift-translation $\cstar$-flow 
 $\{\tauf_t\in \AUTf,\;t\in \R\}$ given above 
  preserves the  grading $\Theta$: 
\begin{equation}
\label{eq:pres-grading}
\Theta
(\tauf_{t}(A))=
\tauf_{t}(\Theta(A))\quad  \text{for all}\  
A\in \Fl \ \text{for each}\ t\in\R.    
\end{equation}
Hence  it induces 
  a $\cstar$-flow 
$\{\tauf_t\in \AUTeven,\;t\in \R\}$
on the even algebra $\Fle$ by restriction.

The time evolution
$\{\tauf_t \in  \AUTf,\;t\in \R\}$   
 \eqref{eq:CARquasifree} on $\Fl$
is  associated to  the  second quantized 
Hamiltonian $\Hinf=\Gamma(\bar{h})$ of the one-particle Hamiltonian $\bar{h}$. 
 In the following, we shall  
 derive  a  concrete  expression  of  $\Hinf$ 
 by  a formal (non rigorous) method.
 Let  
\begin{align}
\label{eq:hck}
\hck:=\sum_{j\in\Z}\cj e^{ijk}, \quad 
\hckcr:=\sum_{j\in\Z}\cjcr e^{ijk}, \  k\in [-\pi, \pi]
\end{align}
Let 
$\Nk:=\hckcr \hck$, namely  the number operator 
  for the  momentum $k\in [-\pi, \pi]$.
In terms of $\Nk$, $\Hinf$ has the  following 
  expression  in the momentum space:
\begin{align}
\label{eq:H-conti}
\Hinf=\int_{-\pi}^{\pi} k \Nk \,dk.  
\end{align}
 Using 
\begin{align*}
[\Nk, \hck] =-i\hck,\
[\Nk, \hckcr] =i\hckcr
\end{align*}
 we  rewrite $\Hinf$  in the lattice space as  
\begin{align*}
\Hinf=
\int_{-\pi}^{\pi}
\sum_{n\in\Z}\sum_{m\in\Z}
k e^{ik(n-m)} \cncr\cm \,dk
=\sum_{n, m\in\Z}
h_{n,m}\cncr \cm,
\end{align*}
 where  for $n,m\in\Z$
\begin{align*}
h_{n,m}=G(n-m)\in\CC,\ G(x):=\int_{-\pi}^{\pi}k e^{ikx} \, dk.
\end{align*}
By  integration by parts 
\begin{align*}
G(x)
&=\left[k \frac{e^{ikx}}{ix}\right]_{-\pi}^{\pi}
 -\int_{-\pi}^{\pi}
 \frac{e^{ikx}}{ix} \, dk=
\pi \frac{e^{i\pi x}+e^{-i\pi x}}{ix}+
\frac{e^{i\pi x}-e^{-i\pi x}}{x^2}\nonum\\
&=
-2\pi i \frac{\cos (\pi x)}{x}+2i \frac{\sin(\pi x)}{x^2}
\end{align*}
By plugging $n-m$ to $x$, we  obtain 
\begin{align}
\label{eq:Onmfinal}
\Hinf=\sum_{n, m\in\Z}
h_{n,m} \cncr \cm,
\ h_{n,m}=\frac{-2\pi i (-1)^{n-m}}{n-m}.
\end{align}

In fact, the  construction   
 of  continuous  shift translations  described above   
 has been   given  in  
 $\S$2.1  of  \cite{GRO},  $\S$6 of \cite{RWW}, 
  Appendix E of \cite{WIL}, and  $\S$III  of \cite{ZIM}.
We, however, emphasize  that 
   the  second-quantization method  for the fermion system 
  seems not valid for  
 the {\emph{infinite}}   quantum spin lattice system $\Als$.
This assertion will be rigorously validated later in Proposition 
 \ref{prop:no-EXT}.

As the  shift-translation $\cstar$-flow  
 $\{\tauf_t\in \AUTf,\;t\in \R\}$  is a  quasi-free dynamics, 
we can track   its  dynamical behavior  as follows:
\begin{prop}
\label{prop:Quasi-free}
The discrete shift-translation automorphism group
  $\{\tauf_i\in \AUTf,\;i\in \Z\}$
on the one-dimensional fermion lattice system    $\Fl$
 can be extended to a quasi-free  $\cstar$-flow 
$\{\tauf_t\in \AUTf,\;t\in \R\}$ on  $\Fl$.
It  has the exact  dynamical  formulas in terms of 
 the sinc function $\sinc x:=\frac{\sin x}{x}$ as follows.
  For  $j\in \Z$ and $t\in\R \setminus \Z$, 
\begin{align}
\label{eq:CARtimesinc}
\tauf_t(\cj)
&=\sum_{l\in\Z} \sinc( \pi(j+t-l)) \cl 
=\sum_{l\in\Z} \frac{(-1)^{j-l}\sin(\pi t)}{\pi(j+t-l)} \cl
\in \Fl, \nonum \\
\tauf_t(\cjcr)
&=\sum_{l\in\Z} \sinc( \pi(j+t-l)) \clcr 
=\sum_{l\in\Z} \frac{(-1)^{j-l}\sin(\pi t)}{\pi(j+t-l)} \clcr \in \Fl.
\end{align}
For  $j\in \Z$ and $k\in \Z$, 
\begin{align}
\label{eq:CARtimesZ}
\tauf_k(\cj)
=\cjk \in \core,\quad
\tauf_k(\cjcr)
= \cjkcr \in \core.
\end{align}
The time-shift translation $\tauf_t\in \AUTf$
  ($\forall t\in\R \setminus \Z$) 
 violates the locality 
 in the sense of Definition 2.2.1 of \cite{Farr}{\rm{:}}\ 
There exist   a strict local operator 
whose   time-translation  is  not a  local operator.
\end{prop}

\begin{proof}
The existence of such   $\cstar$-flow 
$\{\tauf_t\in \AUTf,\;t\in \R\}$ is given in the text, and 
\eqref{eq:CARtimesZ} is obvious. 
As
$\widetilde{U_t \chi_j}(k)=e^{i({j+t})k}$ $(k\in[-\pi, \pi])$
 by \eqref{eq:chiFourier} \eqref{eq:INTERt},  
we have 
\begin{align}
\label{eq:Fourier-shift}
U_t \chi_j (l)=
\frac{1}{2\pi}\int_{-\pi}^{\pi}
  e^{i({j+t})k} e^{-ilk}\,dk  
=\frac{1}{\pi(j+t-l)}\sin(j+t-l)\pi,\ (l\in\Z).
\end{align}
This yields the formula \eqref{eq:CARtimesinc}.
The formula \eqref{eq:CARtimesinc} implies that 
 for any $j\in \Z$, 
neither $\tauf_t(\cj)$ nor 
$\tauf_t(\cjcr)$ belongs to $\core$  for  $t\in \R \setminus \Z$. 
 Actually, for  each $t\in \R \setminus \Z$, 
there are many more  $A\in \core$ such that  $\tauf_t(A)\notin  \core$.
Thus   the locality as stated  is violated
  for  all $t\in \R \setminus \Z$. 
\end{proof}

We shall provide  several   consequences 
 of  Proposition \ref{prop:Quasi-free}.
The first one is related to the following question 
 posed in $\S$ 6.1 of \cite{Farr}:  
Is there a meaningful notion of a ground state
for QCA(quantum cellular automata)?  
A  similar question naturally arises:
How about  equilibrium 
 states for   discrete  automorphism groups?
The following corollary provides  a positive answer to this 
question  
   for the shift-translation automorphism group
  $\{\tauf_i\in \AUTf,\;i\in \Z\}$
on the one-dimensional fermion lattice system    $\Fl$.
(However,  we do not know the answer for 
 $\{\taus_i \in \AUTs,\;i\in \Z\}$.)
\begin{coro}
\label{coro:GROUND}
The continuous shift-translation automorphism 
 group 
$\{\tauf_t\in \AUTf,\;t\in \R\}$ on  $\Fl$
 has 
 a unique KMS state $\vp_{\beta}$ at any  inverse temperature $\beta\in \R$.  
 It has   a unique 
 ground state $\vp_{\infty}$ at $\beta=\infty$ 
and a unique  ceiling state $\vp_{-\infty}$ at $\beta=-\infty$.
\end{coro}

\begin{proof}
Due to  Theorem 1 of \cite{ROCCA}, 
 there is a unique  KMS state $\vp_{\beta}$
 on $\Fl$ 
  for each $\beta\in\R$ with respect to 
the time evolution  $\{\tauf_t\in \AUTf,\;t\in \R\}$. 
 
There is no invariant vector in  $L^2[-\pi, \pi]$ under the 
 action of  $\{U_t\in  \Ul(L^2[-\pi, \pi]);\; t\in \R\}$.  
  Thus by  the  discussion given in 
 Example 5.3.20 \cite{BR}, we see that    a unique 
  ground state  $\vp_{\infty}$ 
  with respect to  $\{\tauf_t\in \AUTf,\;t\in \R\}$ exists.
 Similarly, there is a  unique  ceiling state $\vp_{-\infty}$
 with respect to  $\{\tauf_t\in \AUTf,\;t\in \R\}$. 
\end{proof}

\begin{remark}
\label{rem:KMS}
For each $\beta\in\R$, the quasi-free KMS state $\vp_{\beta}$ 
 is not a Gibbs state, since 
its one-particle density  operator $\exp({-\beta \bar{h}} )$
 on $L^2[-\pi, \pi]$ is not of 
 trace-class, see Prop. 5.2.23  \cite{BR}. 
 By  Theorem 4.3 \cite{PWKMS} the  KMS state $\vp_{\beta}$ 
induces  a ${\rm{III}}_1$ factor,
 since  the   assumption on  the time evolution  of this theorem is 
 obviously satisfied  by the quasi-free time evolution 
$\{\tauf_t\in \AUTf,\;t\in \R\}$.
\end{remark}

\begin{remark}
\label{rem:Ground-Momentum}
As the  second-quantized Hamiltonian 
$\Hinf=
\int_{-\pi}^{\pi}
k \Nk \,dk$  given in   \eqref{eq:H-conti}
generates  the spatial translations on $\Fl$, it  
  may be naturally identified with the  momentum operator  ``$P$''. 
Although the existence of a ground state  for $P$
 on the infinite homogeneous quantum system $\Fl$
 seems to be  unexpected,   
 it should be noted  
 that the ground state $\vp_{\infty}$ and 
 the ceiling state $\vp_{-\infty}$ 
   correspond  to the  fermi sea 
  occupied by all left-(right-) moving particles, respectively.
Of course the identification $\Hinf=P$ is  formal;
 we refer to $\S$4.5  of \cite{M-time}
 which compares
 Hamiltonians and   momentum operators 
 in terms of  spatial and time crystals.
\end{remark}

\begin{remark}
\label{rem:BACKWARD}
 In addition to the  violation of the locality  described 
in Proposition \ref{prop:Quasi-free}, 
causality, which denotes  the one-directional nature of dynamics,  
  is not perfectly  satisfied by  the continuous extension 
 of  shift translations $\{\tauf_t\in \AUTf,\;t\in \R\}$.
 This is because   backward dynamics 
 with respect to the time parameter emerges, as illustrated in \eqref{eq:CARtimesinc}. The  violation of the  causality is also  observed in   
 the   continuous  shift translations 
 on  the finite-quantum spin chains,   
 as  dipicted in Fig. 1 of  \cite{GRO}.
 On the infinitely extended system $\Fl$, however, 
the  violation of  the causality for any   
 shift-translation $\cstar$-flow 
  becomes  negligible 
in a large time scale, 
as we will see later in   Proposition  \ref{prop:AA-universal}.  
\end{remark}

\begin{remark}
\label{rem:LR} 
The Lieb-Robinson bound  \cite{LR} is a key  assumption  of the  
 work \cite{RWW}, whereas 
 we do not require it.
   Lieb-Robinson bounds  have been shown for    power-law interactions
 with their decay rate $\frac{1}{r^\alpha}$ ($\alpha>d$), where $d$
 is the spatial dimension of the quantum lattice system, 
 see \cite{El} \cite{KUWA}  \cite{MAT} \cite{WIL},  particularly
 \cite{GONG} for  free fermion models on lattices.
 On the other hand, the formal Hamiltonian  $\Hinf$ of  the continuous 
 shift translations on the one-dimensional fermion lattice system
 consists of the  two-body  translation invariant  interactions 
with  $\frac{1}{r}$-decay ($\alpha=d=1$) 
 rate, as  given  in   \eqref{eq:Onmfinal}.
Therefore,   
 the quasi-free fermion  model associated to  the continuous-shift 
translations  has such a long range that 
 it  is outside  the scope of   these previous works.
\end{remark}

\section{Non extendability of  continuous  quasi-free shifts to the quantum spin system}
\label{subsec:NONext}
In the preceding section, we provided a $\cstar$-flow  
  that  extends the 
 discrete shift-translation automorphism group
  on the fermion lattice system $\Fl$. 
The following proposition demonstrates   
 that it is not possible to extend this  
quasi-free   $\cstar$-flow  on $\Fl$
 given in  Proposition  \ref{prop:Quasi-free}  
 to the quantum spin lattice system $\Als$.
\begin{prop}
\label{prop:no-EXT}
There is no 
$\cstar$-flow on $\Als$ which 
 extends the discrete shift-translation automorphism group
  $\{\taus_i\in \AUTs,\;i\in \Z\}$  on $\Als$ and 
coincides with the quasi-free  $\cstar$-flow
  $\{\tauf_t\in \AUTf,\;t\in \R\}$  
on  the even subalgebra  $\Alse=\Fle$.
 \end{prop}

\begin{proof}
Let  $\{\taus_t\in \AUTs,\; t\in \R\}$ 
 denote a $\cstar$-flow extension 
 of the discrete shift-translation automorphism group
  $\{\taus_i\in \AUTs,\;i\in \Z\}$
on $\Als$. Suppose that it 
 satisfies the assumption 
\begin{align}
\label{eq:even-coinc}
\taus_t(A)=
\tauf_t(A) \quad \forall A\in  \Alse=\Fle  \quad \forall t\in\R.
\end{align}
In the following, we will prove  non-existence of such
 $\cstar$-flow. 
To this end, we  recall    Lemma 4.2 of  \cite{MatsuiO}, 
 see  also \cite{EVANS-LEWIS}. 
Let $B_2:=\ctwocr+\ctwo$.
For each $t\in \R$,  we define  
\begin{align}
\label{eq:Vt}
V_t:=\taus_t(B_2T )  \tauf_t(B_2)T\in \Fle.
\end{align}
The  inclusion   
$V_t\in \Fle$ as stated above is verified as follows.
By  \eqref{eq:even-coinc} 
 $\taus_t$ and $\tauf_t$ are the same  shift-translation group  
 on  $\Alse=\Fle$,   and hence they automatically
  preserve the grading  $\Thes$ and  $\Theta$, respectively.
  So 
$\taus_t(B_2T )\in \Also$ 
 for $B_2T\in  \Also$,  
and $\tauf_t(B_2)\in \Flo$ for 
 $B_2\in   \Flo$. 
Thus we have $\taus_t(B_2T )\in T\Flo$  and 
 $\tauf_t(B_2)T\in T\Flo$.  
From  the identity $T^2=1$ and the inclusion  
$\Flo\cdot \Flo\subset \Fle$, 
 it follows that  
the product  $\taus_t(B_2T )  \cdot \tauf_t(B_2)T\equiv V_t$ belongs to 
 $\Fle$. 
Let
\begin{align}
\label{eq:thetami}
(\theta_{-} f)_j:=
\begin{cases}
f_j & j \ge 1 \\
-f_j & j \le  0
\end{cases}
\end{align}
 for $f=(f_j)_{j\in \Z}\in l^2(\Z)$.
 The unitary operator $\theta_{-}$ on $l^2(\Z)$
 generates $\Theta_{-}$
 on $\Fl$ given in \eqref{eq:Thetami} by the second-quantization procedure.
For each $t\in \R$, we define 
\begin{align}
\label{eq:wt}
\omega_t:=U_t \theta_{-} U_{-t} \theta_{-} \in
 \Ul(l^2(\Z)), 
\end{align}
 where $U_t\in \Ul(l^2(\Z))$ is given  in \eqref{eq:UtZ}.
In  Lemma 4.2 of  \cite{MatsuiO}  the following 
 relation is given: 
 \begin{align}
\label{eq:AdVt}
{\rm{Ad}}(V_t)=\taus_t \circ \Theta_{-} \circ \taus_{-t}
\circ \Theta_{-} \  \text{on}\  \Fle.
\end{align}
From \eqref{eq:wt} \eqref{eq:AdVt}
 it follows 
that  
 \begin{align}
\label{eq:Vt}
V_t c(f)V_t^{\ast}=c(\omega_t f),\
V_t c^{\ast}(f)V_t^{\ast} =c^{\ast}(\omega_t f).
\end{align}

We now consider 
the   continuous shift operator $U_t$ 
 in the place of  the NESS-dynamics  $u_t$ of the original 
 paper  \cite{MatsuiO}.
Due to  $V_t \in \Fle$, 
each $\beta_t:={\rm{Ad}}(V_t)$ is  an  inner automorphism of $\Fl$. 
Accordingly,   $1-\omega_t$  is  a Hilbert Schmidt operator
 on $l^2(\Z)$, see \cite{ARCAR}. 
 This 
  entails   the following boundedness condition. 
\begin{align}
\label{eq:YUU}
I_t \equiv \sum_{m\in\NN} \sum_{n\in\NN}
\left|
(e_{-m}, U_t e_n)
\right|^2<\infty, \quad t\in \R.
\end{align}

 Next, we will show that
the required boundedness  \eqref{eq:YUU}  is not satisfied.
 For each $t\in \R$, we compute 
\begin{align}
\label{eq:YUUkobetu}
(e_{-m}, U_t e_n)=
\frac{1}{2\pi}\int_{-\pi}^{\pi} e^{-imk} e^{-itk} e^{-ink} \, dk
=\frac{1}{2\pi}
\left[ \frac{e^{i(-m-t-n)k}}{i(-m-t-n)} \right]_{-\pi}^{\pi}
\nonumber \\
=\frac{1}{2\pi i(-m-t-n)}
\left( e^{i\pi (-m-n)}e^{-it\pi}
-e^{i\pi (m+n)}e^{it\pi}
\right)
=\frac{(-1)^{-m-n-1}\sin(\pi t)}{\pi (-m-t-n)}.
\end{align}
Thus, we obtain 
\begin{align}
\label{eq:form}
I_t=
\frac{1}{\pi^2 }
\sum_{m\in\NN} \sum_{n\in\NN} \frac{|\sin(\pi t)|^2}{ |(m+t+n)|^2}.
\end{align}
Hence, unless  $\sin(\pi t)=0$, this  double infinite sum diverges, namely, 
\begin{equation}
\label{eq:mugen}
I_t= \infty \quad \forall t \in \R \setminus \Z.
\end{equation}
 Due to   the inconsistency between
 \eqref{eq:YUU} and \eqref{eq:mugen}, 
   the existence of the inner automorphism 
 ${\rm{Ad}}(V_t)$ is negated. Therefore,  
  the assumed  $V_t\in \Fl$ actually does not exist  
 for all non-integer real numbers; 
 of course,  integers   $t\in\Z$ are  excluded due to  the existence of 
 the discrete shift-translation automorphism group
  $\{\taus_i\in \AUTs,\;i\in \Z\}$.
We  conclude that   
  the  $\cstar$-flow 
$\{\taus_t\in \AUTs,\;t\in \R\}$
satisfying the  assumption  does not exist.
\end{proof}

\section{Asymptotically Abelian condition}
\label{sec:COMMON}
Proposition \ref{prop:Quasi-free} (together with Remark\ref{rem:BACKWARD}) 
demonstrates  that 
 the locality and causality (one-sided direction) 
 satisfied by the discrete 
 shift-translation automorphism group of $\Fl$
  are lost by  its continuous extension. 
We ask   what dynamical properties are shared 
 by both the discrete shift-translation automorphism group 
 and its   continuous extensions. 
We will  show that the asymptotically abelian condition 
 serves as an example of such shared properties. 

The asymptotically abelian condition 
  is  a general  characterization of  quantum dynamics 
 that ensures delocalization  \cite{BR}.
It implies  that every  local operator 
 will escape   from   its original position,  
 and never come back there after some time.  

We  shall recall the  precise  definition 
 of the asymptotically abelian condition \cite{DKK67} in the following.
First, we  give the definition for  the case $\Als$.
  Consider a $\cstar$-flow $\{\alpha_t \in  \AUTs,\;t\in \R\}$ on $\Als$.
We say that it  satisfies   
the  (norm) asymptotically abelian condition, 
 if  the following asymptotic formula  holds:
\begin{equation}
\label{eq:asym}
\lim_{|t| \to \infty} [A, \;  \alpha_t(B)] =0
\quad \text{for every} \  A, B \in \Als, 
 \end{equation}
where the  norm convergence is used.
Similarly,   consider a $\Z$-automorphism group 
$\{\tilde{\alpha}_i \in  \Als,\;i\in \Z\}$. 
It  satisfies   the  (norm) asymptotically abelian condition if
\begin{equation}
\label{eq:asym-i}
\lim_{|i| \to \infty} [A, \;  \tilde{\alpha}_i(B)] =0  
\quad \text{for every} \  A, B \in \Als,  
 \end{equation}
where the  norm convergence is used. 
We then move to  the case $\Fl$,  
 which  has 
 the graded structure \eqref{eq:CAREO} 
 with the   grading  $\Theta$. 
 The   graded commutator  on $\Fl$ is defined  by
 the  following mixture of the commutator and the anti-commutator: 
\begin{align}
\label{eq:gcom}
[\Ap, \;  B]_{\Theta} &:= [\Ap, \;  B] \quad
\text{for}\quad \Ap \in \Ale, \ B \in \Al, \nonum\\ 
[\Am, \; \Bm]_{\Theta} &:= \{\Am,  \; \Bm\} \quad 
\text{for}\quad \Am, \Bm \in \Alo.
\end{align}
A $\cstar$-flow
$\{\alpha_t \in  \AUTf,\;t\in \R\}$ on the fermion system $\Fl$ 
  satisfies   the  $\Theta$-graded  (norm)
 asymptotically abelian condition, if 
\begin{equation}
\label{eq:Gasym}
\lim_{|t| \to \infty} [A, \;  \alpha_t(B)]_{\Theta} =0  
\quad \text{for every} \  A, B \in \Fl,  
 \end{equation}
where the  norm convergence is used.
Similarly,
 $\{\tilde{\alpha}_i \in  \AUTf,\;i\in \Z\}$ 
satisfies  the  (norm) $\Theta$-graded asymptotically abelian condition, if 
\begin{equation}
\label{eq:Gasym-i}
\lim_{|i| \to \infty} [A, \;  \tilde{\alpha}_i(B)]_{\Theta} =0  
\quad \text{for every} \  A, B \in \Fl,
 \end{equation}
where the  norm convergence is used. 
In what follows, we shortly  call   
 the   asymptotically abelian condition 
    `AA-condition'
 and  the  $\Theta$-graded asymptotically abelian condition
   `$\Theta$-AA-condition'.

As a typical example, the discrete shift-translation automorphism group
  $\{\taus_i\in \AUTs,\;i\in \Z\}$
on   $\Als$ satisfies  AA-condition, whereas
the discrete shift-translation automorphism group
  $\{\tauf_i\in \AUTf,\;i\in \Z\}$
on  $\Fl$ satisfies  $\Theta$-AA-condition.
If a $\cstar$-flow 
 $\{\alpha_t \in  \AUTs\ {\rm{or}} \in \AUTf,\;t\in \R\}$ satisfies
 AA  condition or  $\Theta$-AA-condition, 
 then  so does its restricted   
$\Z$-automorphism group 
  $\{\alpha_i \in  \AUTs \ {\rm{or}} \in \AUTf,\;i\in \Z\}$. 
The following proposition  
 demonstrates that   the converse implication holds in general.

\begin{prop}
\label{prop:AA-universal}
Let $\{\alpha_t\in \AUTs,\;t\in \R\}$ be any 
$\cstar$-flow on $\Als$ that extends 
the discrete shift-translation automorphism group
  $\{\taus_i\in \AUTs,\;i\in \Z\}$. 
Then  $\{\alpha_t\in \AUTs,\;t\in \R\}$  satisfies 
 AA-condition.
Let $\{\alpha_t\in \AUTf,\;t\in \R\}$ be any 
$\cstar$-flow on $\Fl$ that extends 
the discrete shift-translation automorphism group
  $\{\tauf_i\in \AUTf,\;i\in \Z\}$. 
Then  $\{\alpha_t\in \AUTf,\;t\in \R\}$  satisfies 
$\Theta$-AA-condition.
\end{prop}

\begin{proof}
We consider the case  $\Als$. 
We  introduce some  convenient notation.
 Each  $t\in \R_{+}$ 
 is  uniquely written  as the  sum of 
 its non-negative integer part $\tint \in \NN \cup{0}$ 
and its decimal part $\tdec\in [0, 1)$.
 For example,  we have $\tint=3$ and $\tdec=0.14159$ for  $t=3.14159$.  
 Let $\Ks$ be some  subset of $\Als$.
 Let $\varepsilon>0$ and $A\in \Als$.
 If $\displaystyle{\inf_{X\in \Ks}\Vert A-X \Vert<\varepsilon}$, then we  write  $A \ineps  \Ks$.
 For  $\I\Subset\Z$, its translation by $x\in\Z$
 is denoted as $\I+x:=\{y+x\in \Z; \; y\in \I\}$.
 For  $N\in \NN\cup\{0\}$ define the 
 following discrete segment centered at  $0\in\Z$ 
\begin{equation*}
[N]\equiv \{-N, -N+1, \cdots, -1, 0, 1, 
\cdots, N-1, N\}
\end{equation*}
We  denote 
\begin{equation*}
\AlsN\equiv 
\Als\left([N]\right). 
\end{equation*}
For $x\in \Z$,  
let 
\begin{equation*}
\AlsNx\equiv
\taus_x( \AlsN)
\end{equation*}

 To show the   formula 
 \eqref{eq:asym} it is  enough
 to derive   
\begin{equation}
\label{eq:asymII}
\lim_{t \to +\infty} [A, \;  \alpha_t(B)] =0
\quad \text{for each pair of} \ A, B \in \cores. 
 \end{equation}
For any given   $A, B \in \cores$, 
 there is a minimum  $N\in \NN$
 such that  $A, B \in \AlsN$.
Let  $B(t):=\alpha_t(B)\in \Als$, which is  
an $\Als$-valued norm continuous function for  $t\in \R$.
 We consider $B(t)$ on the closed segment  
$[0, 1]\equiv \{t\in\R ; \;0\le t \le 1 \}$.
For $p\in\NN$,  take the following 
 finite subset of $\Als$:
\begin{equation}
\Bl_{p}:= \{B(0), B(1/p), B(2/p),  \cdots, B(k/p),  \cdots,  B(1-1/p),
 B(1)\}, 
\quad (0\le  k \le p,\ k\in \NN).
\end{equation}
For any fixed $\varepsilon>0$, 
 there exists a sufficiently large $p \in \NN$ such that 
\begin{equation}
\label{eq:Bt-eps}
B(t)\ineps\Bl_p
 \quad  \text{for every } \ t\in [0, 1],
\end{equation}
 Since 
$\Bl_p$ is a finite subset of 
$\Als$, there exists a sufficiently large 
$M(>N)\in \NN$ such that 
\begin{equation}
\label{eq:Bseg-eps}
\Bl_{p} \ineps \AlsM.
\end{equation}
From \eqref{eq:Bt-eps} \eqref{eq:Bseg-eps}
we have 
\begin{equation}
\label{eq:Bt-Am}
B(t) \inDoubeps \AlsM
\quad  \text{for every } \ t\in [0, 1].
\end{equation}

We will crudely estimate  
  the location  of $\alpha_t(B)$ for large $t\in\R$.
By definition,  we have 
\begin{equation}
\label{eq:Bt-bunkai}
B(t)=B(\tint+\tdec)=\alpha_{\tint}\left(B(\tdec)\right)
=\taus_{\tint}\left(B(\tdec)\right)
\quad \text{for each } \ t(=\tint+\tdec) \in \R.
\end{equation}
From \eqref{eq:Bt-Am}
and  \eqref{eq:Bt-bunkai}
 it follows that 
\begin{equation}
\label{eq:Bt-IDOU}
B(t) \inDoubeps \AlsidouM
\quad  \text{for each} \ t\in \R.
\end{equation}
 Hence,  the intersection 
$[N]\cap [M]+\tint$ will  eventually become empty 
 as  $t\to+\infty$ ($\tint \to+\infty$).
Thus,  due to  the local commutativity  of $\Als$,  
 we have shown  \eqref{eq:asymII}.

The proof for \eqref{eq:asymII} provided  above  
 demonstrates   that 
 the location  of $\alpha_t(B)$ is near to the shift-translation 
 $\taus_{\tint}(B)$ 
 and becomes  far away 
from the original  location of the local operator $B$ 
as the time  $t\in\R$ proceeds. 
 This characteristic also holds  true in the case of $\Fl$.
Thus,  we can show    
\begin{equation}
\label{eq:GasymII}
\lim_{t \to \infty} [A, \;  \alpha_t(B)]_{\Theta} =0  
\quad \text{for each pair of} \  A, B \in \core,   
 \end{equation}
repeating  the same argument for \eqref{eq:asymII}.
 We need  not  consider 
 distinctions arising from the graded structure of the fermion system $\Fl$.
\end{proof}

\begin{remark}
\label{rem:ROH}
This is a side remark. 
The  Rohlin property \cite{KIS} is
 another   property which becomes invalid 
 by the continuous extension of discrete shift translations.
The shift-translation automorphism group 
  $\{\tauf_i\in \AUTf,\;i\in \Z\}$ is known to have 
the  Rohlin property  \cite{BKRS}, whereas 
its  quasi-free extension  
$\{\tauf_t\in \AUTf,\;t\in \R\}$ fails to satisfy it  because 
 of the existence of  KMS states Corollary \ref{coro:GROUND}. 
\end{remark}

\section{Discussion}
\label{sec:DIS}
 We specified   
the formal Hamiltonian $\Hinf$
for  the  shift-translation quasi-free $\cstar$-flow   
 $\{\tauf_t\in \AUTf,\;t\in \R\}$ given in Proposition  \ref{prop:Quasi-free}. 
 It consists of the  two-body  translation invariant  interactions 
with  $\frac{1}{r}$-decay.  Let us   recapitulate its  formula  given 
in \eqref{eq:Onmfinal}: $\Hinf
=\sum_{n, m\in\Z}\frac{-2\pi i (-1)^{n-m}}{n-m} \cncr \cm$.  
 
We note that  its   derivation  is not rigorous;   
 the   meaning  of $\Hinf$ 
  in  $\cstar$-theory \cite{BR} remains  unclear.
More precisely, we do not know whether 
 and how it is associated to  a pre-generator 
of the  $\cstar$-flow   
 $\{\tauf_t\in \AUTf,\;t\in \R\}$.
Although 
the quasi-free  $\cstar$-flow $\{\tauf_t\in \AUTf,\;t\in \R\}$
 has   KMS  states  
as shown in Corollary \ref{coro:GROUND},
the one-site energy, which is  required for  
 the description of equilibrium states (attaining the minimum free-energy),  
 seems   not obvious.
At least, its existence  does not follow from 
 the  general theory  of   fermion lattice systems 
\cite{EQ} \cite{BR}.
We conjecture that  the   alternative $\pm$ sign  
of each creation-annihilation pair interaction 
 may  provide  a clue for the   cancellation of long-range effects,  
  yielding   
 a well-defined  $\cstar$-flow 
on  the infinite  fermion lattice system. 
(There are  some long-range  quantum lattice models 
 which do not have   $\cstar$-flows, see \cite{THIRBCS} \cite{BRU2}.)

The  above quasi-free  $\cstar$-flow on $\Fl$ 
  is  a special  example of  continuous extensions 
 of discrete shift translations.
 Proposition  \ref{prop:no-EXT}   merely excludes 
 one specific extension method from $\Fl$ to $\Als$; it is not a 
 thorough no-go statement.  There may be   other   
  continuous  extensions  of discrete shift translations
  on $\Fl$ and  $\Als$.
We shall raise  the following questions.

\medskip

\noindent{\bf{Question A}}:\  
Is  there  
another  (non-quasi-free) $\cstar$-flow  on $\Fl$ that 
 extends  the  discrete  shift translations?

\medskip

\noindent{\bf{Question B}}:\  
Is there  
a $\cstar$-flow on $\Als$ that 
 extends  the  discrete  shift translations?

\medskip

\noindent{\bf{Question C}}:\  
Is there  
a $\cstar$-flow on $\Als$ satisfying AA condition?
\medskip

If  a positive  example  for {\bf{Question A}}
exists on $\Fl$, then we will argue 
 whether   its  extension   
  to $\Als$  is possible or not.
  Of course, by   Proposition  \ref{prop:AA-universal},  
if a positive example of {\bf{Question B}} is found, then 
it gives a positive  example for {\bf{Question C}} as well.
 
\section*{Acknowledgments}
We would like to    thank 
Jean-Bernard Bru, Toshihiko Masuda and Itaru Sasaki 
 for  correspondence.
This work was financially supported by   Kakenhi (grant no.
21K03290).

\end{document}